\documentclass[12pt]{llncs}
\usepackage{graphicx}
\usepackage{algorithm,algorithmic}
\usepackage[figuresright]{rotating}
\usepackage{amsmath,amssymb,amsfonts}
\usepackage{bm}
\usepackage{inputenc}
\usepackage{amsfonts}
\begin{document}
\title{New Competitive Semi-online Scheduling Algorithms for Small Number of Identical Machines\thanks{Supported by Veer Surendra Sai University of Technology, Burla, Odisha, 768018 INDIA.}}
%
%
\author{Debasis Dwibedy\and
Rakesh Mohanty}
%
%
\institute{Veer Surendra Sai University of Technology, Burla, Odisha 768018, INDIA 
\email{(debasis.dwibedy, rakesh.iitmphd)}@gmail.com}
%
\maketitle              
\begin{abstract}
Design and analysis of constant competitive deterministic semi-online algorithms for the multi-processor scheduling problem with small number of identical machines have gained significant research interest in the last two decades. In the semi-online scheduling problem for makespan minimization, we are given a sequence of independent jobs one by one in order and upon arrival, each job must be allocated to a machine with prior knowledge of some \textit{Extra Piece of Information (EPI)} about the future jobs. Researchers have designed multiple variants of semi-online scheduling algorithms with constant competitive ratios by considering one or more \textit{EPI}. In this paper, we propose four new variants of competitive deterministic semi-online algorithms for smaller number of identical machines by considering two \textit{EPI} such as $Decr$ and $Sum$. We obtain improved upper bound and lower bound results on the competitive ratio for our proposed algorithms, which are comparable to the best known results in the literature. In two identical machines setting with known $Sum$, we show a tight bound of $1.33$ on the competitive ratio  by considering a sequence of equal size jobs. In the same setting we achieve a lower bound of $1.04$ and an upper bound of $1.16$ by considering $Sum$ and a sequence of jobs arriving in order of decreasing sizes. For three identical machines setting with known $Decr$ and $Sum$, we show a lower bound of $1.11$ on the competitive ratio. In this setting, we obtain an upper bound of $1.5$ for scheduling a sequence of equal size jobs and achieves an upper bound of $1.2$ by considering a sequence of decreasing size jobs. Further we develop an improved competitive algorithm with an  upper bound of $1.11$ on the competitive ratio. 
\end{abstract}
\section{Introduction} 
The classical scheduling problem requires irrevocable scheduling of a list of $n$ jobs on a set of $m (\geq 2)$ identical parallel machines with an objective to minimize the completion time of all jobs i.e., makespan. This problem with $m=2$ was proved to be NP-complete in [1] by a polynomial time reduction from the well-known partition problem. The above hardness result was achieved by using the concept of splitting a weighted set of items into two subsets of equal weights.\\\\
\textit{Offline Scheduling.} When the list of jobs to be scheduled is given at the outset, the problem is called \textit{offline}. The algorithm designed for the offline $m$ machines scheduling problem is known as an \textit{offline algorithm}. Out of all offline algorithms, the one that achieves the smallest makespan is called the \textit{optimal offline algorithm $OPT$}. \textit{Longest Processing Time (LPT)} [2] is one of the primitive offline algorithms for the $m$ machine offline scheduling problem. Algorithm \textit{LPT} first sorts a given list of jobs in order of non-increasing sizes, then follows the order to schedule the jobs one by one on the machine that has the current minimum load. In worst case, algorithm \textit{LPT} achieves a makespan which is $1.16$ times of the optimum makespan.\\\\
\textit{Online and Semi-online Scheduling.} In practice jobs are not known at the outset, rather they are revealed one by one in order. A job $J_i$ must be scheduled on a machine $M_j$ as soon as it is received with no clue on the successive jobs $J_{i+1}$, where $1\leq i<n$.  Such a scheduling problem is known as \textit{online scheduling}. A recent survey on the variants of the online scheduling problem can be found in [3]. In online scheduling, when some \textit{Extra Piece of Information (EPI)} about future jobs is given a priori, the problem becomes \textit{semi-online} [4-5].  \\
In this article, we are interested in semi-online scheduling of a sequence of $n$ independent jobs on $m$ identical parallel machines settings, where we know jobs arrive in order of non-increasing sizes and the value of the total sum of their sizes. We are constrained to schedule each job irrevocably and non-preemptively. Our objective is to minimize the makespan.\\\\
\textit{Practical Significance.} Semi-online scheduling problem often arises in many practical applications such as resource management in production systems, interactive parallel computations, robot navigation, distributed data processing and clients requests management by the network servers. In a client-server model, let us consider a video download application, which allows a client to download videos from the internet. When a client places a request to download a particular video file, the application knows a priori the total size of the requested file, associated data packets and their order of delivery. These information, if given to the network routers, would certainly help to optimize the congestion in the networks and to minimize the delivery time at the destination. Efficient semi-online scheduling algorithms can be extremely useful in minimizing the bandwidth consumption at any link of a network, while routing the data packets across several networks. In certain applications, where a limited number of resources are used, the design and analysis of efficient semi-online algorithms are of great theoretical and research  significance.\\\\
\textit{Performance Measure.} The efficiency of a semi-online algorithm $ALG$ is evaluated by competitive analysis method [6]. In this method a competitive ratio is computed based on the cost obtained by the semi-online algorithm $ALG$ and the cost obtained by its optimal offline counterpart. Let us consider an input instance $\sigma=\langle J_1, J_2, \ldots, J_n \rangle$, consisting of $n$ items of a cost minimization problem $X$. Let $C_{ALG}$ and $C_{OPT}$ be the costs incurred by $ALG$ and $OPT$ respectively. The competitive ratio can be defined as the smallest $b (\geq 1)$ for which $\frac{C_{ALG}}{C_{OPT}}\leq b$ for any permutations in the arrival order of the items in $\sigma$.\\\\
\textit{Research Motivation.}  The ongoing research in semi-online scheduling has been predominantly influenced by two factors. The first one accounts for the improvement in the best known competitive bounds and the second one inclines towards the exploration of practically significant new \textit{EPI}. The \textit{EPI} is completely application dependent. The \textit{EPI} that helps in achieving the current best bound for a particular setting might not be available in all real world online scheduling applications. Therefore analysis of semi-online algorithms based on unexplored \textit{EPI} or using a combination of multiple \textit{EPI} to achieve better competitive analysis results is a major research challenge. According to our knowledge, a maiden work [7] has been reported till date in the literature on semi-online scheduling with known $Decr$ and $Sum$. The existing literature lacks an extensive study on semi-online scheduling of equal size jobs on two and three identical machines settings. In our study we propose four new variants of  constant competitive deterministic semi-online scheduling algorithms for smaller number of machines with two \textit{EPI}.\\\\
\textit{Our Contribution.} We investigate the semi-online scheduling problem with known $Decr$ and $Sum$.  We study the non-preemptive variants for makespan minimization in two and three identical parallel machines settings. Two natural input job sequence patterns are explored based on the information on $Decr$. The first pattern is equivalent to a sequence of equal size jobs and the second one resembles to a sequence of decreasing size jobs. We prove the lower bound results for the studied settings by characterizing the input job sequences. For two identical machines setting, a lower bound of $1.33$ on the competitive ratio is shown by considering both input patterns, and a lower bound $1.04$ is proved by considering only the second input pattern. We propose a deterministic semi-online algorithm named \textit{2DS}, which achieves an upper bound of $1.33$ on the competitive ratio for both input patterns by maintaining a near equal load on each machine. We improve the upper bound to $1.16$ by considering only the second input pattern and proposing a deterministic semi-online algorithm named,  \textit{Improved 2DS (I2DS)}. Algorithm \textit{I2DS} always keeps a machine with heavy load as compare to the other one. The imbalance in loads between the machines leads to an improved competitive bound. However, algorithm \textit{2DS} remains best possible for the first input pattern as the load imbalance principle does not help to improve the upper bound. For three identical machines setting with known $Decr$ and $Sum$, we show a lower bound of $1.11$ on the competitive ratio, which holds for both input patterns. In this setting, we propose a semi-online algorithm named \textit{3DS}, which achieves an upper bound $1.5$ in scheduling equal size jobs. Furthermore, we prove that algorithm \textit{3DS} is $1.2$ competitive with known $Sum$ and $Decr$, where jobs arrive in order of decreasing sizes. Finally we improve the upper bound $1.2$ on the competitive ratio to $1.11$ by considering only the second input pattern and proposing the algorithm \textit{Improved 3DS (I3DS)}. Our proposed algorithms are efficient as they hold the defined bounds for all input instances of the studied settings.\\\\
\textit{Organization.} The remaining sections of this paper is organized as follows. Section 2 discusses some basic terminologies and seminal contributions related to our study. Section 3 presents our proposed algorithms and associated competitive analysis results for non-preemptive semi-online scheduling for makespan minimization in two and three identical parallel machines settings with \textit{EPI} $Decr$ and $Sum$. Finally section 4 provides the concluding remarks and highlights the future scope of our work.  
\section{Preliminaries and State-of-the-art Results}
We present the basic terminologies, notations and definitions related to our work in Table \ref{tab:Basic Terms Notations and Definition}.
\begin{table}[!htbp]
\centering
\caption[centre]{Basic Terminologies Notations and Definitions}
\begin{tabular} {ccp{5.9cm}}
\hline
\textbf{Terminology} & \textbf{Notations} &\textbf{Definitions} \\
\hline
Job  & $J_i$ &  An executable unit, where $1\leq i\leq n$. \\
Machine & $M_j$ & A processing unit, where $1\leq j\leq m$.\\
Processing Time or Size  &  $p_{i}$ & Time taken by a job $J_i$ to execute on a machine $M_j$.\\
Completion Time & $c_i$ & The time at which a job $J_i$  finishes its execution\\
Load & $l_j$ & Sum of  sizes of the jobs that have been assigned to a machine $M_j$. \\
Makespan & $C_{max}$ & $\max\{l_j|1\leq j\leq m\}$. \\
Total sum of sizes of all jobs & Sum & $\sum_{i=1}^{n}{p_i}$.\\
Largest size & Max or $p_{max}$ & $\max \{p_i| 1\leq i\leq n\}$.\\
Non-increasing job sizes & Decr & $p_{i+1}\leq p_i$, for $i\geq 1$.\\
Tightly grouped processing time & TGRP & $p_i\in [a, b]$, where $a, b> 0$.\\
Optimum Makespan & $C_{OPT}$ & $\max\{\frac{1}{m}\cdot \sum_{i=1}^{n}{p_i},\hspace*{0.2cm}  p_{max}\}$.\\
\hline
\end{tabular}
\label{tab:Basic Terms Notations and Definition}
\end{table} 
\subsection{State-of-the-art Results} 
In a semi-online variant of the online scheduling problem, besides the current job some \textit{EPI} on the future jobs is also given. The consideration of different \textit{EPI} leads to different semi-online variants. The ongoing research in semi-online scheduling focuses on the improvement in the existing best competitive bounds for each of the variants. We particularly concentrate on makespan minimization in two and three identical parallel machines settings. We present an overview of the important competitive analysis results as follows.\\
\textbf{Two Identical Machines.} Kellerer et al. [4] studied a semi-online version of the well-known partition problem, which is equivalent to semi-online scheduling on two identical machines with known $Sum$. By considering $Sum=2$, they proposed a $1.33$ competitive semi-online algorithm. The algorithm schedules a sequence of jobs in such a manner that the maximum load incurred on a machine is at most $\frac{2}{3}\cdot Sum$. 
Seiden et al. [5] claimed that algorithm \textit{LPT} with an upper bound of $1.16$ on the competitive ratio is best possible for semi-online scheduling with known $Decr$.\\
Many researchers investigated the semi-online scheduling problem by considering multiple \textit{EPI} and significantly improved the existing best competitive bounds.  Angelelli [8] introduced a pair-wise combination of two \textit{EPI} such as $Sum$ and lower bound on job's processing time. A tight bound of $1.33$ on the competitive ratio was achieved. Tan and He [7] considered $Sum$ and $Max$ as the known \textit{EPI} and obtained a tight bound of $1.2$ on the competitive ratio.\\
They also considered $Sum$ and $Decr$ as the known \textit{EPI} to achieve a lower bound of $1.11$ on the competitive ratio. To prove their claim they considered two specific instances $\sigma_1$, $\sigma_2$ with six jobs in each instance and $Sum=18$, where $\sigma_1$ considers $p_1=p_2=4$, $p_3=p_4=p_5=p_6=2.5$ and $\sigma_2$ considers $p_1=p_2=4$, $p_3=p_4=p_5=3$ and $p_6=1$. Although the  achieved bound holds for the considered instances, it does not help to trace out all instances those account for the  lower bound result.  A generalized lower bound of the problem as a factor of the known $Sum$ is required to capture all such instances. They proposed the algorithm \textit{SD} and achieved a matching upper bound on the competitive ratio by considering $Sum=18$. Algorithm \textit{SD} schedules the first job $J_1$ on machine $M_1$. The second job $J_2$ is assigned to $M_1$ and the remaining jobs are scheduled on $M_2$ if $l_1+p_2\leq \frac{5}{9}\cdot Sum$. If $\frac{7}{18}\cdot Sum\leq l_1+p_2< \frac{4}{9}\cdot Sum$ or $l_1+p_2> \frac{5}{9}\cdot Sum$, then $J_2$ is scheduled on $M_2$ and the remaining jobs $J_i$, where $i\geq 3$ are scheduled based on the following two stopping criteria. Criterion 1: If $l_j+p_i\leq \frac{5}{9}\cdot Sum$, then $J_i$ is assigned to such a $M_j$ and the remaining jobs are scheduled on the other machine, where $j\in \{1, 2\}$. Criterion 2: If $l_j+p_i> \frac{5}{9}\cdot Sum$, $\forall j$, then $J_i$ is scheduled on the $M_j$ for which the current load $l_j=\min\{l_1, l_2\}$ and the remaining jobs are assigned to the other machine. At any moment if the loads $l_1=l_2$, then $J_i$ is assigned to $M_2$. Although algorithm \textit{SD} improves the existing competitive bound, it is does not cover all instances of the problem. For example, let us consider a sequence of three equal size jobs with $p_1=p_2=p_3=1$ and another sequence of three decreasing size jobs with $p_1=5$, $p_2=4$ and $p_3=3$. For the first and the second sequence any deterministic semi-online algorithm has a competitive ratio of at least $1.33$ and $1.16$ respectively. Therefore the semi-online scheduling problem on two identical machines with known $Decr$ and $Sum$ requires further investigations in the design of optimal semi-online algorithms for $n\geq 3$.\\ Angelelli [9] improved the $1.33$ tight bound to $1.2$ by considering $Sum$ and upper bound ($ub$) on job's processing time, where $ub\in (0.5, 0.6)$. For $ub\in (0.75, 1)$, Angelelli achieved a tight bound of $1+\frac{ub}{3}$ on the competitive ratio. In [10], Angelelli considered $Sum$ and various intervals of the upper bound on job's processing time as the known \textit{EPI}. For $ub\in [\frac{1}{b}, \frac{2(b+1)}{b(2b+1)}]$ a tight bound $1+\frac{1}{2b+1}$ and for $ub\in (\frac{2b-1}{2b(b-1)}, \frac{1}{b-1}]$ a tight bound $(\frac{b-1}{3})ub+0.66(\frac{b+1}{b})$ were achieved on the competitive ratio, where $b\geq 2$. Cao et al. [11] obtained a tight bound of $1.2$ by considering $Max$ and the value of the optimum makespan as the known \textit{EPI}.
Cao and Wan [12] considered the known information on $Decr$ and $TGRP(1, r)$ and proved a tight bound of $1.16$ on the competitive ratio. We now present a summary of the best known competitive bounds  in Table \ref{tab:Best Known Competitive Bounds for Two Identical Machines}.   
\begin{table}[!htbp]
\centering
\caption[centre]{Best Known Competitive Bounds for Two Identical Machines}
\begin{tabular} {ccc}
\hline
\textbf{EPI} & \textbf{Lower Bound} &\textbf{Upper Bound} \\
\hline
$Sum$  & $1.33$ &  $1.33$ \\
$Decr$ & $1.16$ & $1.16$\\
$Sum, TGRP(lb)$  &  $1.33$ &  $1.33$\\
$Sum, Max$ & $1.2$ & $1.2$ \\
$Sum, Decr$ & $1.11$ & $1.11$ \\
$Sum, TGRP(ub)$ & $1.2$ & $1.2$ \\
$Max, Opt$ & $1.2$ & $1.2$\\
$Decr, TGRP(1, r)$ & $1.16$ & $1.16$\\
\hline
\end{tabular}
\label{tab:Best Known Competitive Bounds for Two Identical Machines}
\end{table}\\
\textbf{Three Identical Machines.} He and Dosa [13] initiated a study on semi-online scheduling in three identical machines settings. They considered $TGRP(r)$ as the known \textit{EPI} and achieved an upper bound $1.5$ on the competitive ratio for $r\in (2, 2.5)$. Furthermore, they obtained an upper bound $\frac{4r+2}{2r+3}$ for $r\in (2.5, 3)$. Angelelli [14] achieved a lower bound $1.392$ and an upper bound $1.421$ on the competitive ratio by considering known $Sum$. \\
Hua et al. [15] introduced multiple \textit{EPI} in semi-online scheduling on three identical machines settings. They considered $Sum$ and $Max$ as the known \textit{EPI} and obtained a lower bound $1.33$ and an upper bound $1.4$ on the competitive ratio. For this setting, Wu et al. [16] improved the lower bound of Hua e al. [15] to achieve a tight bound $1.33$ on the competitive ratio.\\
By considering only known $Decr$, Cheng et al. [17] achieved a tight bound $1.18$, which is the best competitive bound known till date in semi-online scheduling on three identical machines settings.   We present a summary of the best known competitive bounds  in Table \ref{tab:Best Known Competitive Bounds for Three Identical Machines}.
\begin{table}[h]
\centering
\caption[centre]{Best Known Competitive Bounds for Three Identical Machines}
\begin{tabular} {ccc}
\hline
\textbf{EPI} & \textbf{Lower Bound} &\textbf{Upper Bound} \\
\hline
$TGRP(r)$ & --- & $1.5$\\ 
$Sum$  & $1.392$ &  $1.421$ \\
$Sum, Max$ & $1.33$ & $1.33$\\
$Decr$  &  $1.18$ &  $1.18$\\
\hline
\end{tabular}
\label{tab:Best Known Competitive Bounds for Three Identical Machines}
\end{table} \\
The current state-of-the-art results on the competitive ratio  for other variants of the semi-online scheduling problem can be found  in a recent survey of Epstein [18].
\section{Our Results on Semi-online Scheduling with Known Decr and Sum}
In the semi-online scheduling problem  with $Decr$ and $Sum$,  we know that $p_{i}\geq p_{i+1}$ for $1\leq i < n$ and we are given the value of $\sum_{i=1}^{n}{p_i}$ at the outset. For this problem, we analyze the competitiveness of any deterministic semi-online algorithm \textit{ALG} by considering two practically significant \textit{input job sequence patterns}, i.e., $I_1$ and $I_2$ based on the known $Decr$. As we know $Decr$ means, $n$ jobs arriving such that $p_1\geq p_2\geq p_3\geq \ldots \geq p_n$. We now can distinguish the input patterns $I_1$ and $I_2$ as follows.\\\\
\hspace*{1.2cm}  \textit{$I_1$ :} $p_1=p_2=p_3= \ldots = p_n$, i.e., $p_{i+1}=p_i$, $\forall i$.\\
\hspace*{1.2cm} \textit{$I_2$ :} $p_1>p_2>p_3> \ldots > p_n$, i.e., $p_{i+1} < p_i $, $\forall i$.\\\\
The input pattern $I_1$ often arises in a client server environment, where for an instance a client requests only for the static web pages or several clients request for the home page of a website on the fly. The web server processes such requests by taking almost equal time for each of them. The input pattern $I_2$ is natural in the context of product manufacturing and parallel computations, where a larger size task is splitted into several smaller size jobs and the jobs are given one by one in order of decreasing sizes to a central scheduler for the assignment of scarce resources. The objective is to yield a product or a result within a minimum time.   
Any of the other possibilities on the $Decr$ must be a subset of the union of $I_1$ and $I_2$. In fact, in $I_1$, the information on $Decr$ is meaningless. The problem can also be interpreted as semi-online scheduling of a sequence of equal size jobs with known $Sum$. While the problem with $I_2$ can be considered as semi-online scheduling of a sequence of decreasing size jobs with known $Sum$. We study the problem with respect to $I_1$ and $I_2$ in two and three identical parallel machines settings as follows.    
\subsection{Two Identical Parallel Machines ($P_2$)}
When the objective is to minimize the makespan $C_{max}$, we denote the problem as $P_2|Decr, Sum|C_{max}$. Let $C_{ALG}$ and $C_{OPT}$ be the makespans incurred by any deterministic semi-online algorithm $ALG$ and the optimal offline algorithm $OPT$ respectively. The lower bound on the competitive ratio of the problem $P_2|Decr, Sum|C_{max}$ represents the smallest $b(\geq 1)$ such that $\frac{C_{ALG}}{C_{OPT}}\geq b$ for all instances of the problem. We show the lower bound results of the problem as follows.\\\\ 
\subsubsection{Lower Bound Results}
\begin{theorem}
Any deterministic semi-online algorithm $ALG$ for the problem $P_2|Decr, Sum|C_{max}$ has a competitive ratio of at least $1.33$. 
\end{theorem}
\begin{proof}
We prove the theorem by adversary method. Let $Sum=k(> 6)$ is known and we are given a sequence of three jobs $J_1$, $J_2$ and $J_3$, where $p_1=\frac{k+3}{3}$ and $p_2+p_3=\frac{2k-3}{3}$. We consider the following two cases based on the assignment of initial two jobs.\\\\
\textit{Case 1.  $J_1$ and $J_2$ are assigned to the same machine.}\\ Let us consider $p_2=p_3=\frac{2k-3}{6}$. By assigning $J_3$ on the other machine, we have $C_{ALG}\geq \frac{4k+3}{6}$, while $C_{OPT}=\frac{k}{2}$. Therefore, \\
\hspace*{2.7cm} $\frac{C_{ALG}}{C_{OPT}}\geq \frac{4}{3}+\frac{1}{k}=1.33+\frac{1}{k}$ \hspace*{4.2cm} (1) \\\\
\textit{Case 2.  $J_1$ and $J_2$ are assigned to two different machines.}\\ Let us consider $p_2=\frac{k}{3}$ and $p_3=\frac{k-3}{3}$.\\
\textit{Sub-case 2.1.} $J_3$ is assigned to the machine, where job $J_1$ has been assigned. We now have $C_{ALG}\geq \frac{2k}{3}$, while $C_{OPT}= \frac{k}{2}$. This implies \\ 
\hspace*{2.7cm} $\frac{C_{ALG}}{C_{OPT}}\geq \frac{4}{3}=1.33$ \hspace*{5.7cm} (2) \\\\
\textit{Sub-case 2.2.} $J_3$ is assigned to the machine, where job $J_2$ has been assigned. We now have $C_{ALG}\geq \frac{2k-3}{3}$, while $C_{OPT}=\frac{k}{2}$. This implies \\ 
\hspace*{2.5cm} $\frac{C_{ALG}}{C_{OPT}}\geq \frac{4}{3}-\frac{2}{k} \rightarrow 1.33$ as $k\rightarrow \infty$\hspace*{3.3cm} (3) \\\\
By Eqs. (1), (2) and (3), we can conclude that there exists an instance of the problem $P_2|Decr, Sum|C_{max}$  such that $\frac{C_{ALG}}{C_{OPT}}\geq 1.33$. \hfill\(\Box\)
\end{proof} 
It can be observed that in the proof of Theorem 1, we have considered the critical instances based on the input patterns $I_1$ and $I_2$ with minimum number of jobs to show all possible assignments. We have obtained the lower bound of the problem by considering the minimal makespan, which was incurred by the best assignment policy.  The result reflects that  the largest value of $k$ leads to a generalized lower bound of the problem on the competitive ratio. However, we can further minimize the lower bound by considering only $I_2$ and the least value of $k$ i.e., $k=7$. We now present the lower bound of the problem by considering only $I_2$ as follows. 
\\\\
\begin{theorem}
Any deterministic semi-online algorithm $ALG$ for the problem  $P_2|Decr, Sum|C_{max}$ with $I_2$ has a competitive ratio of at least $1.04$.
\end{theorem}
\begin{proof}
To prove our claim let us consider $Sum=k(\geq 7)$ and a sequence of three jobs $J_1$, $J_2$ and $J_3$, where $p_1=\frac{12k}{25}$ and $p_2+p_3=\frac{13k}{25}$. Let us consider the best assignment policy and the critical instance as follows. The jobs $J_1$ and $J_2$ are scheduled on two different machines and the job $J_3$ is assigned to the machine where $J_2$ has already been assigned. For $p_2=\frac{7k}{25}$ and $p_3=\frac{6k}{25}$, we have $C_{ALG}\geq \frac{13k}{25}$, while $C_{OPT}=\frac{k}{2}$. This implies, $\frac{C_{ALG}}{C_{OPT}}\geq \frac{26}{25}=1.04$. \hfill\(\Box\)
\end{proof}
\subsubsection{Algorithm 2DS}
In the problem $P_2|Decr, Sum|C_{max}$, we know larger size jobs arrive upfront in the sequence and the total sum of sizes of all jobs beforehand. We aim to schedule the jobs in such a manner that the load of one of the machines is always at most $\frac{1}{2}\cdot Sum$. We propose a deterministic semi-online algorithm named \textit{2DS} for the problem by considering the input patterns $I_1$ and $I_2$. Let $C_{2DS}$ be the makespan incurred by algorithm \textit{2DS}. The algorithm works as follows. 
\begin{algorithm}
\caption{2DS}
\begin{algorithmic}
\scriptsize
\STATE Initially, $l_1=l_2=0$, $Sum=\sum_{i=1}^{n}{p_i}$, and $p_{i+1}\leq p_i$ for $1\leq i< n$ \\
\STATE WHILE a new job $J_{i}$ is given with known $Decr$ and $Sum$ DO\\
\STATE \hspace*{0.2cm} BEGIN\\
\STATE \hspace*{0.5cm} IF $l_1+p_i\leq \frac{1}{2}\cdot Sum$  \\
\STATE \hspace*{0.8cm} THEN assign job $J_i$ to machine $M_1$ \\
\STATE \hspace*{0.8cm} UPDATE $l_1=l_1+p_i$\\
\STATE \hspace*{0.5cm} ELSE \\
\STATE \hspace*{0.8cm} assign job $J_i$ to machine $M_2$ \\
\STATE \hspace*{0.8cm} UPDATE $l_2=l_2+p_i$\\
\STATE \hspace*{0.5cm} $i=i+1$\\
\STATE \hspace*{0.2cm} END\\
\STATE Return \hspace*{0.3cm} $C_{2DS}=\max\{l_1, l_2\}$
\end{algorithmic}
\end{algorithm}\\
The upper bound on the competitive ratio of our proposed algorithm $2DS$ for the problem $P_2|Decr, Sum|C_{max}$ represents the smallest $b(\geq 1)$ such that $\frac{C_{2DS}}{C_{OPT}}\leq b$ for all job sequences of the problem. We now show the upper bound results of algorithm \textit{2DS} with  respect to $I_1$ and $I_2$ by Theorem 3 and Theorem 4 respectively as follows.
\subsubsection{Upper Bound Result}
\begin{theorem}
For all job sequences  of the problem $P_2|Sum|C_{max}$ with $I_1$, where $p_i=x$, $x\geq 1$, $n\geq 3$ and $1\leq i\leq n$, we have 
$\frac{C_{2DS}}{C_{OPT}}\leq 1.33$.
\end{theorem} 
\begin{proof}
We consider the following two cases based on the number of jobs $n$. \\
\textit{Case 1. If $n$ is even.}\\
 Let us consider $n=2a$ for any $a\geq 2$. As we know that $p_i=x$ for $1\leq i\leq n$, we have $Sum=2ax$, implies, $C_{OPT}= ax$. As $n$ is even, algorithm \textit{2DS} assigns the initial $\frac{n}{2}=a$ jobs to machine $M_1$ and incurs $l_1=ax$. The remaining $a$ jobs are scheduled on machine $M_2$ to obtain $l_2=ax$. We now have $C_{2DS}= ax$, this implies, $\frac{C_{2DS}}{C_{OPT}}= 1$. \\
\textit{Case 2. If $n$ is odd.}\\
Let us consider $n=2a+1$ for any $a\geq 1$. We have $Sum=(2a+1)x$, implies,  $C_{OPT}= \frac{(2a+1)x}{2}$. Algorithm \textit{2DS} assigns the initial $a$ jobs to machine $M_1$ and incurs load $l_1=ax < \frac{(2a+1)x}{2}$. The remaining $a+1$ jobs are scheduled on machine $M_2$ to obtain a load $l_2=(a+1)x$. As $l_2>l_1$, we now have $C_{2DS}\leq (a+1)x$. Therefore, we have $\frac{C_{2DS}}{C_{OPT}}\leq \frac{2x(a+1)}{x(2a+1)}\leq \frac{2a+2}{2a+1}\leq \frac{4}{3}=1.33$, for $n\geq 3$ and $a\geq 1$. \hfill\(\Box\)
\end{proof} 
\begin{theorem}
For all job sequences  of the problem $P_2|Decr, Sum|C_{max}$ with $I_2$, we have $\frac{C_{2DS}}{C_{OPT}}\leq 1.33$.
\end{theorem}
\begin{proof}
We prove the theorem by critical case analysis method. Let $p_{max}$ be the largest job, we have $p_{max}=p_1$. Let us  consider the following two cases based on the size of the first job.\\
\textit{Case 1. If $p_1 > \frac{1}{2}\cdot Sum$.} \\
Then $\sum_{i=2}^{n}{p_i} < \frac{1}{2}\cdot Sum$, implies $C_{OPT}=p_1$. Algorithm \textit{2DS} assigns $J_1$ to machine $M_2$ and incurs a load  $l_2=p_1$. The remaining jobs are assigned to machine $M_1$, which incurs a load $l_1 < \frac{1}{2}\cdot Sum < p_1$. Implies, $C_{2DS}=l_2=p_1$. Therefore, $\frac{C_{2DS}}{C_{OPT}}=1$.\\
\textit{Case 2. If $p_1\leq \frac{1}{2}\cdot Sum$.}\\
\textit{Sub-case 2.1. If $p_1=\frac{1}{2}\cdot Sum$.}\\
 Then, $\sum_{i=2}^{n}{p_i} = \frac{1}{2}\cdot Sum$. Algorithm \textit{2DS} assigns $J_1$ to machine $M_1$ and incurs a load $l_1=\frac{1}{2}\cdot Sum$. The remaining jobs are scheduled on machine $M_2$, which incurs a load $l_2=\frac{1}{2}\cdot Sum$. Thus, $C_{2DS}=l_1=l_2=\frac{1}{2}\cdot Sum$, while $C_{OPT}=\frac{1}{2}\cdot Sum$. Therefore, \hspace*{4.0cm}$\frac{C_{2DS}}{C_{OPT}}=1$.\\\\
\textit{Sub-case 2.2. If $p_1 < \frac{1}{2}\cdot Sum$.}\\
Then algorithm \textit{2DS} schedules $J_1$ on machine $M_1$. Here, the following two cases arise.\\
\textit{Sub-sub-case 2.2.1. No other job is assigned to machine $M_1$.}\\ Remaining jobs are scheduled on machine $M_2$. We now have two possibilities.\\
\textit{Possibility 1. If $\sum_{i=2}^{n-1}{p_i} > \frac{1}{2}\cdot Sum$.}\\ Then $p_1+p_n < \frac{1}{2}\cdot Sum$, implies there exists a job $J_n$ which can be assigned to machine $M_1$, which is a contradiction to the Sub-sub-case 2.2.1. Hence, Possibility 1 does not occur.\\
\textit{Possibility 2. If $\sum_{i=2}^{n-1}{p_i}\leq \frac{1}{2}\cdot Sum$.    }\\
Then before the scheduling of $J_n$, load $l_2 \leq \frac{1}{2}\cdot Sum$. However, after scheduling of $j_n$, we have the updated load $l_2=l_2+p_n > \frac{1}{2}\cdot Sum$ as $l_1 < \frac{1}{2}\cdot Sum$. As $l_1 +p_n > \frac{1}{2}\cdot Sum$, without loss of generality, we can bound the value of $p_n$ as $2\leq p_n\leq p_{max}-2$ for $n\geq 3$. We consider the following critical instance of the problem to show the largest competitive ratio. Let us consider an instance with three jobs $J_1$, $J_2$ and $J_3$, where $p_1=p_{max}$, $p_2=p_{max}-1$ and $p_3=p_{max}-2$. Algorithm \textit{2DS} now schedules job $J_1$ on machine $M_1$ and assigns remaining jobs to machine $M_2$ to incur $l_1=p_{max}$ and $l_2=2p_{max}-3$. As $l_2 > l_1$, we now have $C_{2DS}\leq 2p_{max}-3$, while $C_{OPT}\geq \frac{3}{2}(p_{max}-1)$. Therefore, we have \\\\
\hspace*{0.8cm} $\frac{C_{2DS}}{C_{OPT}}\leq \frac{2(2p_{max}-3)}{3p_{max}-3}\leq \frac{4p_{max}-6}{3p_{max}-3}\leq \frac{4}{3}$, for $p_{max}\geq 3$ \hspace*{2.3cm}(4)\\\\
\textit{Sub-sub-case 2.2.2. At least one job other than $J_1$ is assigned to machine $M_1$.}\\
Let $J_k$ be the latest job, which has been assigned to  machine $M_1$ and let $l^{'}_{1}$ be the load of $M_1$ just before the scheduling of $J_k$. We now have \\\\
\hspace*{1.8cm} $l^{'}_{1}+p_k\leq \frac{1}{2}\cdot Sum > l_1$ \hspace*{6.1cm}(5)\\\\ 
If $J_k=J_n$, then $l^{'}_{1}+p_k = \frac{1}{2}\cdot Sum$ or $l^{'}_{1}+p_k < \frac{1}{2}\cdot Sum$. If $l^{'}_{1}+p_k = \frac{1}{2}\cdot Sum$, then $C_{2DS}=\frac{1}{2}\cdot Sum$, while $C_{OPT}=\frac{1}{2}\cdot Sum$, implies, \hspace*{3.1cm} $\frac{C_{2DS}}{C_{OPT}}=1$.\\\\
If $l^{'}_{1}+p_k < \frac{1}{2}\cdot Sum$, then let $l^{'}_{2}$ be the load of machine $M_2$ at this point of time. By Eq. (5), it is clear that $l^{'}_{2}<l_2$. Therefore, Eq. (4) holds for this case as well. We now can conclude that $\frac{C_{2DS}}{C_{OPT}}\leq \frac{4}{3}=1.33$.  \hfill\(\Box\)  
\end{proof} 
For the problem $P_2|Decr, Sum|C_{max}$ with $I_1$, the algorithm \textit{2DS} is best possible in the sense that no semi-online algorithm can overcome the critical instance with three jobs, where $p_1=p_2=p_3=1$ to obtain a better competitive ratio than $1.33$. However, we improve the upper bound $1.33$ on the competitive ratio to $1.16$ for the problem with $I_2$, i.e.,  $p_{i+1}< p_i$, $\forall i$ by proposing a deterministic semi-online algorithm named \textit{Improved 2DS}.
\subsubsection{Algorithm Improved 2DS (\textit{I2DS})}
schedules an incoming job $J_i$ on machine $M_1$ as long as the load $l_1+p_i\leq \frac{7}{12}\cdot Sum$. A job $J_x$, which incurs a load $p_x$ on machine $M_1$ such that $l_1+p_x > \frac{7}{12}\cdot Sum$ is assigned to machine $M_2$. All such $J_x$ are scheduled on machine $M_2$. The idea is to maintain loads on the machines $M_1$ and $M_2$ such that $l_1\leq \frac{7}{12}\cdot Sum$ and $\frac{5}{12}\leq l_2< \frac{7}{12}\cdot Sum$. For a better understanding of the competitive analysis of algorithm \textit{I2DS}, let us normalize the processing times of the jobs in such a manner that $Sum=12$.
\begin{algorithm}
\caption{I2DS}
\begin{algorithmic}
\scriptsize
\STATE Initially, $l_1=l_2=0$, $Sum=\sum_{i=1}^{n}{p_i}$, and $p_{i+1}< p_i$ for $1\leq i< n$ \\
\STATE WHILE a new job $J_{i}$ is given with known $Decr$ and $Sum$ DO\\
\STATE \hspace*{0.2cm} BEGIN\\
\STATE \hspace*{0.5cm} IF $l_1+p_i\leq \frac{7}{12}\cdot Sum$  \\
\STATE \hspace*{0.8cm} THEN assign job $J_i$ to machine $M_1$ \\
\STATE \hspace*{0.8cm} UPDATE $l_1=l_1+p_i$\\
\STATE \hspace*{0.5cm} ELSE \\
\STATE \hspace*{0.8cm} assign job $J_i$ to machine $M_2$ \\
\STATE \hspace*{0.8cm} UPDATE $l_2=l_2+p_i$\\
\STATE \hspace*{0.5cm} $i=i+1$\\
\STATE \hspace*{0.2cm} END\\
\STATE Return \hspace*{0.3cm} $C_{I2DS}=\max\{l_1, l_2\}$
\end{algorithmic}
\end{algorithm}
\subsubsection{Upper Bound Result}
\begin{theorem}
For all job sequences of the problem $P_2|Decr, Sum|C_{max}$ with $I_2$, where $n\geq 3$ and $p_n\geq 3$, we have $\frac{C_{I2DS}}{C_{OPT}}\leq 1.16$.
\end{theorem}
\begin{proof}
When no job is assigned to any of the machines, we have loads $l_1=l_2=0$. It is given that $Sum=12$. This implies, \\
\hspace*{2.5cm} $C_{OPT}\geq 6$ \hspace*{7.0cm}(6)
We consider the following two critical cases based on the size of the first job $J_1$.\\
\textit{Case 1.} If $p_1>\frac{7}{12}\cdot Sum$.\\
Let us consider $A=\sum_{i=2}^{n}{p_i}$. It is clear that $A< \frac{5}{12}\cdot Sum$. Algorithm \textit{I2DS} now schedules job $J_1$ on machine $M_2$ and assigns the remaining jobs $J_i$ to machine $M_1$, where $2\leq i\leq n$. This implies, $l_2=p_1> \frac{7}{12}\cdot Sum$ and $l_1=A< \frac{5}{12}\cdot Sum$. As $l_2> l_1$, we now have $C_{I2DS}=p_1$, while $C_{OPT}=p_1$. \\
\textit{Case 2.} If $\frac{5}{12}\cdot Sum < p_1\leq \frac{7}{12}\cdot Sum$.\\
Then $l_1+p_1>\frac{5}{12}\cdot Sum$. Clearly, $l_1< \frac{5}{12}\cdot Sum$ as $l_1=0$ initially. Algorithm \textit{I2DS} schedules job $J_1$ on machine $M_1$ and the updated load $l_1\leq \frac{7}{12}\cdot Sum$. We now have  $l_2\leq \frac{7}{12}\cdot Sum$ as $l_1> \frac{5}{12}\cdot Sum$. This implies, $C_{I2DS}\leq \frac{7}{12}\cdot Sum$. Thus, we obtain from Eq. (6), $\frac{C_{I2DS}}{C_{OPT}}\leq \frac{7}{6}=1.16$. The bound holds for any other cases as well.  \hfill\(\Box\)
\end{proof}
\subsection{Three Identical Parallel Machines ($P_3$)}
When three identical parallel machines are given in the semi-online scheduling problem with known $Decr$ and $Sum$ to minimize the makespan ($C_{max}$), we denote the problem as $P_3|Decr, Sum|C_{max}$. Let $C_{ALG}$ and $C_{OPT}$ be the makespans obtained by any semi-online algorithm $ALG$ and the optimal offline algorithm $OPT$ respectively. We present the lower bound of the problem $P_3|Decr, Sum|C_{max}$ as follows. 
\subsubsection{Lower Bound Result}
\begin{theorem}
Any deterministic semi-online algorithm $ALG$ for the problem $P_3|Decr, Sum|C_{max}$ has a competitive ratio of at least $1.11$.
\end{theorem}
\begin{proof}
We follow the adversary method to prove our claim. Let $Sum=27$ is known in advance and we are given a sequence of four jobs $J_1$, $J_2$, $J_3$, $J_4$, where $p_1=9$ and $p_2+p_3+p_4=18$. We assume that at least one job must be scheduled on each machine by the end of the job sequence. We now can consider the following three cases based on the assignment of  initial two jobs. \\
\textit{Case 1.  $J_1$ and $J_2$ are assigned to the same machine $M_j$.}\\
Let us consider $M_j=M_1$ and $p_2=6$, implies, $l_1=15$. We are now left with jobs $J_3$ and $J_4$ such that $p_3+p_4=12$. Jobs $J_3$ and $J_4$ must be scheduled on two different machines other than $M_1$. By considering $p_3=p_4=6$, we obtain the updated loads  $l_2=l_3=6$. Thus, $C_{ALG}\geq 15$, while $C_{OPT}=9$. Therefore, we have\\
\hspace*{2.7cm}$\frac{C_{ALG}}{C_{OPT}}\geq \frac{15}{9}=1.66$ \hspace*{5.8cm}   (7)\\
\textit{Case 2. $J_1$ and $J_2$ are scheduled on different machines.}\\
\textit{Sub-case 2.1. $J_3$ is assigned to the machine, where $J_1$ has already been assigned.}\\
Let us consider $J_1$ has been assigned to machine $M_1$. We have jobs $J_2$ and $J_3$ such that $p_2+p_3\leq 17$. We now have ten possibilities for $(p_2, p_3)$  i.e., $(9, 8)$, $(9, 7)$, $(9, 6)$, $(9, 5)$, $(8, 8)$, $(8, 7)$, $(8, 6)$, $(8, 5)$, $(7, 7)$ and $(7, 6)$. To minimize the load on machine $M_1$, let us consider $p_2=9$ and $p_3=5$. Thus, we have the updated loads $l_1=14$, $l_2=9$ and $l_3=0$. Now job $J_4$ with $p_4=4$ must be scheduled on machine $M_3$ to update $l_3=4$. Thus, we have $C_{ALG}\geq 14$, while $C_{OPT}=9$. Therefore, we have\\
\hspace*{2.7cm}$\frac{C_{ALG}}{C_{OPT}}\geq \frac{14}{9}=1.55$ \hspace*{5.8cm}   (8)\\
\textit{Sub-case 2.2.} $J_3$ is assigned to the machine, where $J_2$ has already been assigned.\\
Let us consider job $J_2$ has been assigned to machine $M_2$. Now we have jobs $J_2$ and $J_3$ such that $p_2+p_3\leq 17$. Thus, we have three possibilities for $(p_2, p_3)$  i.e., $(8, 5)$, $(7, 6)$ and $(6, 6)$. 
Let us consider $p_2=p_3=6$. We now have the updated loads $l_1=9$, $l_2=12$ and $l_3=0$. Again $J_4$ with $p_4=6$ must be scheduled on the machine $M_3$ and incurs a load $l_3=6$. Thus, we have $C_{ALG}\geq 12$, while $C_{OPT}=9$. Therefore, we have \\
\hspace*{2.7cm}$\frac{C_{ALG}}{C_{OPT}}\geq \frac{12}{9}=1.33$ \hspace*{5.7cm}   (9)\\\\
\textit{Sub-case 2.3. $J_3$ is scheduled on a different machine.}\\
Let us consider the following four instances, where $J_i/p_i$ denote a job and its processing time.\\\\
\hspace*{2.5cm} $\sigma_1=\langle J_1/9, J_2/8, J_3/8, J_4/2 \rangle$\\
\hspace*{2.5cm} $\sigma_2=\langle J_1/9, J_2/8, J_3/7, J_4/3 \rangle$\\ 
\hspace*{2.5cm} $\sigma_3=\langle J_1/9, J_2/8, J_3/6, J_4/4 \rangle$\\
\hspace*{2.5cm} $\sigma_4=\langle J_1/9, J_2/8, J_3/5, J_4/5 \rangle$\\\\
It can be easily shown with each of the above instances that any deterministic semi-online algorithm $ALG$ achieves a competitive ratio such that \\
\hspace*{2.7cm}$\frac{C_{ALG}}{C_{OPT}}\geq \frac{10}{9}=1.11$\hspace*{5.7cm}(10)\\\\
Therefore, by Eqs. (7), (8), (9) and (10), we conclude that Theorem 6 holds true. \hfill\(\Box\)
\end{proof}
\textbf{Algorithm 3DS}\\\\
We now propose a deterministic semi-online algorithm named \textit{3DS} for the problem $P_3|Decr, Sum|C_{max}$ by considering both $I_1$ and $I_2$. Let $C_{3DS}$ be the makespan obtained by algorithm \textit{3DS}. The algorithm works as follows. 
\begin{algorithm}
\caption{3DS}
\begin{algorithmic}
\scriptsize
\STATE Initially, $l_1=l_2=l_3=0$, $Sum=\sum_{i=1}^{n}{p_i}$, and $p_{i+1}\leq p_i$ for $1\leq i< n$ \\
\STATE WHILE a new job $J_{i}$ arrives with known $Decr$ and $Sum$ DO\\
\STATE \hspace*{0.2cm} BEGIN\\
\STATE \hspace*{0.5cm} IF $l_1+p_i\leq \frac{1}{3}\cdot Sum$  \\
\STATE \hspace*{0.8cm} THEN assign job $J_i$ to machine $M_1$ \\
\STATE \hspace*{0.8cm} UPDATE $l_1=l_1+p_i$\\
\STATE \hspace*{0.5cm} ELSE \\
\STATE \hspace*{0.8cm} Select  a machine $M_j$, where $j\in\{2, 3\}$ for which the current load $l_j=\min\{l_2, l_3\}$ \\
\STATE \hspace*{0.8cm} Assign job $J_i$ to the machine $M_j$
\STATE \hspace*{0.8cm} UPDATE $l_j=l_j+p_i$\\
\STATE \hspace*{0.5cm} $i=i+1$\\
\STATE \hspace*{0.2cm} END\\
\STATE Return \hspace*{0.3cm} $C_{3DS}=\max\{l_1, l_2, l_3\}$
\end{algorithmic}
\end{algorithm}\\
Algorithm \textit{3DS} is an improved variant of Graham's list scheduling algorithm [19]. The objective of algorithm \textit{3DS} is to always keep the load of machine $M_1$ such that $l_1\leq \frac{1}{3}\cdot Sum$ and to maintain loads on machines $M_2$ and $M_3$ such that the loads  $l_2$ and $l_3$ are at most $1.5(C_{OPT})$.\\\\ 
\textbf{Upper Bound Results}\\\\
We analyze and present the performance of algorithm \textit{3DS} with respect to $I_1$ and $I_2$ as follows.\\\\
\begin{theorem}
For all job sequences  of the problem $P_3|Decr, Sum|C_{max}$ with $I_1$, where $p_i=x$, $x\geq 1$, $n\geq 4$ and $1\leq i\leq n$, we have $\frac{C_{3DS}}{C_{OPT}}\leq 1.5$.
\end{theorem} 
\begin{proof}
We consider the following three critical cases based on the number of jobs $n$.\\
\textit{Case 1.} If $n=3a$, for any $a\geq 2$.\\
Then the initial $a$ jobs are assigned to machine $M_1$ and update the load $l_1=ax$. Algorithm \textit{3DS} assigns the remaining $2a$ jobs with load $2ax$ to machines $M_2$ and $M_3$ such that exactly $a$ jobs are scheduled on each of the machines and incurs loads $l_2=l_3=ax$. Thus, we have $C_{3DS}=ax$, while $C_{OPT}=ax$. Therefore, we have \\
\hspace*{3.0cm}$\frac{C_{3DS}}{C_{OPT}}=1$ \hspace*{6.5cm} (11)  \\\\
\textit{Case 2.} If $n=3a+1$ for $a\geq 1$.\\
The initial $a$ jobs are scheduled on machine $M_1$ to incur a load $l_1=ax$. Out of the remaining $2a+1$ jobs, algorithm \textit{3DS} assigns the first $2a$ jobs on machines $M_1$ and $M_2$ such that $l_2=l_3=ax$. Now the last job $J_n$ can either be scheduled on machine $M_2$ or on machine $M_3$ to incur a load either $l_2=(a+1)x$ or $l_3=(a+1)x$. Thus, we have $C_{3DS}\leq (a+1)x$, while $C_{OPT}= \frac{(3a+1)x}{3}$. Therefore, we have\\ \hspace*{2.5cm}
$\frac{C_{3DS}}{C_{OPT}}\leq \frac{3(a+1)}{3a+1}\leq \frac{3}{2}=1.5$ \hspace*{4.5cm}(12)\\\\
\textit{Case 3.} If $n=3a+2$ for $a\geq 1$.\\
Then again the initial $a$ jobs are scheduled on machine $M_1$ to update the load $l_1=ax$. Algorithm \textit{3DS} schedules the next $2a$ jobs on machines $M_1$ and $M_2$ to update the loads $l_2=l_3=ax$ (i.e., exactly $a$ jobs are assigned to each of the machines). Out of the remaining two jobs, exactly one job is assigned to $M_2$ and $M_3$ respectively to update the loads $l_2=l_3=(a+1)x$. Thus, we have $C_{3DS}\leq (a+1)x$, while $C_{OPT}= \frac{(3k+2)x}{3}$. Therefore, we have\\ \hspace*{2.5cm}           $\frac{C_{3DS}}{C_{OPT}}\leq \frac{3(a+1)}{3a+2}\leq \frac{6}{5}=1.2$ \hspace*{4.5cm}(13)\\\\
By Eqs. (11), (12) and (13), we can conclude that Theorem 7 holds true. \hfill\(\Box\)  
\end{proof}
\begin{theorem}
For all job sequences  of the problem $P_3|Decr, Sum|C_{max}$ with $I_2$, where $n\geq 4$ and $1\leq i\leq n$, we have $\frac{C_{3DS}}{C_{OPT}}\leq 1.2$.
\end{theorem}
\begin{proof}
We prove the theorem by method of contradiction. Let us assume that the input job sequence $\sigma = \langle J_1, J_2, \ldots, J_t \rangle$ is the smallest instance of the problem with respect to number of jobs that contradicts Theorem 8. Thus $t$ is minimal. Let $C_{3DS}(\sigma)$ and $C_{OPT}(\sigma)$ be the makespans incurred by algorithm \textit{3DS} and $OPT$ respectively for the sequence $\sigma$. Suppose $J_r$ is a job in $\sigma$ that finishes last in the schedule, generated by algorithm \textit{3DS} and $r< t$. We now have an instance $\sigma_1 = \langle J_1, J_2, \ldots, J_r \rangle$ such that $C_{3DS}(\sigma_1)=C_{3DS}(\sigma)$, while $C_{OPT}(\sigma_1)\leq C_{OPT}(\sigma)$. \\\\
This implies,  $\frac{C_{3DS}\sigma_1}{C_{OPT}\sigma_1} \geq \frac{C_{3DS}(\sigma)}{C_{OPT}(\sigma)} > 1.2$.\\\\
Thus the instance $\sigma_1$ constitutes the smallest counterexample to Theorem 8, which contradicts our assumption on the minimality of $t$. Therefore the instance $\sigma$ with $t$ jobs forms the smallest counterexample such that $\frac{C_{3DS}(\sigma)}{C_{OPT}(\sigma)}> 1.2$.\\
Implies, for any other instance $\sigma_2$ with $t-1$ jobs, we have $\frac{C_{3DS}(\sigma_2)}{C_{OPT}(\sigma_2)}\leq 1.2$, where $t-1\geq 4 $. Algorithm \textit{3DS} always maintains the loads on the machines $M_1$, $M_2$ and $M_3$ such that $l_1\leq l_2$ and $l_3\leq l_2$. \\
Let us consider the instance $\sigma_2 = \langle J_1, J_2, \ldots, J_{t-1} \rangle$. After scheduling all jobs of $\sigma_2$, algorithm \textit{3DS} updates the loads of $M_1$, $M_2$ and $M_3$ such that $l_1\leq \frac{1}{3}\cdot Sum$, $l_2\leq \frac{2}{5}\cdot Sum$ and $\frac{4}{15}\cdot Sum \leq l_3 < \frac{2}{5}\cdot Sum$. It is clear that assignment of the job $J_t$ makes $\frac{C_{3DS}(\sigma)}{C_{OPT}(\sigma)}> 1.2$. We know that before scheduling of job $J_t$, we have $l_3< l_2$ and $l_3\leq \frac{2}{5}\cdot Sum-1$. Hence job $J_t$ is scheduled on machine $M_3$ and the updated load $l_3> \frac{2}{5}\cdot Sum$. Without loss of generality, we can bound the processing time of job $J_t$ as $2\leq p_t\leq p_{max}-4$. This implies,\\\\
\hspace*{1.2cm} $C_{3DS}(\sigma)=l_3+p_{max}-4> \frac{2}{5}\cdot Sum$.\\\hspace*{1.2cm} $\frac{2}{5}\cdot Sum+p_{max}-5> \frac{2}{5}\cdot Sum$ \\\\
implies, \hspace*{0.7cm}$p_{max}> 5$. \\\\
To show the correctness of our assumptions, let us consider the smallest instance $\sigma_3 = \langle J_1/6, J_2/5, J_3/4, J_4/3, J_5/2 \rangle$ based on our assumptions on $p_{max}$, $p_t$ and $\sigma$. After scheduling all jobs of $\sigma_3$, algorithm $3DS$ updates the loads such that $l_1=6$, $l_2=l_3=7$. Hence $C_{3DS}(\sigma_3)\leq 7$, while $C_{OPT}(\sigma_3)= \frac{20}{3}$. Therefore $\frac{C_{3DS}(\sigma_3)}{C_{OPT}(\sigma_3)}\leq \frac{21}{20}< 1.2$. The instance $\sigma_3$ is a counterexample to our assumptions on $p_{max}$, $p_t$ and $\sigma$. Hence there does not exist any counterexample to Theorem 8.  Therefore the Theorem holds true. \hfill\(\Box\)
\end{proof}
\textbf{Algorithm Improved 3DS}\\\\
For the problem $P_3|Decr, Sum|C_{max}$ with $I_1$, the algorithm \textit{3DS} is best possible in the sense that no semi-online algorithm can overcome the instance $\sigma = \langle J_1/1, J_2/1, J_3/1, J_4/1 \rangle$ to obtain a better competitive ratio than $1.5$. However, the upper bound $1.2$ for the problem $P_3|Decr, Sum|C_{max}$ with $I_2$ can be improved to $1.11$.  We now propose a more sophisticated, yet improved deterministic semi-online algorithm than \textit{3DS} for the problem $P_3|Decr, Sum|C_{max}$ with $I_2$. We name the algorithm as \textit{Improved 3DS (I3DS)} and it works as follows.
\begin{algorithm}
\caption{I3DS}
\begin{algorithmic}
\scriptsize
\STATE Initially, $l_1=l_2=l_3=0$, $Sum=\sum_{i=1}^{n}{p_i}=27$, and $p_{i+1}< p_i$ for $1\leq i< n$ \\
\STATE WHILE a new job $J_{i}$ is given with known $Decr$ and $Sum$ DO\\
\STATE \hspace*{0.2cm} BEGIN\\
\STATE \hspace*{0.5cm} IF $l_1+p_i\leq \frac{1}{3}\cdot Sum$  \\
\STATE \hspace*{0.8cm} THEN assign job $J_i$ to machine $M_1$ \\
\STATE \hspace*{0.8cm} UPDATE $l_1=l_1+p_i$\\
\STATE \hspace*{0.5cm} ELSE IF $l_1+p_i\leq \frac{10}{27}\cdot Sum$  \\
\STATE \hspace*{0.8cm}
THEN assign job $J_i$ to machine $M_2$ \\
\STATE \hspace*{0.8cm} UPDATE $l_2=l_2+p_i$\\
\STATE \hspace*{0.5cm}
ELSE\\
\STATE \hspace*{0.8cm}
assign job $J_i$ to machine $M_3$ \\
\STATE \hspace*{0.8cm} UPDATE $l_3=l_3+p_i$\\
\STATE \hspace*{0.5cm} $i=i+1$\\
\STATE \hspace*{0.2cm} END\\
\STATE Return \hspace*{0.3cm} $C_{I3DS}=\max\{l_1, l_2, l_3\}$
\end{algorithmic}
\end{algorithm}\\
The aim of algorithm \textit{I3DS} is to always maintain loads on machines $M_1$, $M_2$ and $M_3$ such that $l_1\leq \frac{9}{27}\cdot Sum$, $l_2\leq \frac{10}{27}\cdot Sum$ and $\frac{8}{27}\cdot Sum \leq l_3\leq  \frac{10}{27}\cdot Sum$. For a better understanding of the competitive analysis of algorithm \textit{I3DS}, let us normalize the processing times of the jobs in such a manner that $Sum=27$.
\subsubsection{Upper Bound Result}
\begin{theorem}
For all job sequences  of the problem $P_3|Decr, Sum|C_{max}$ with $I_2$, where $n\geq 4$ and $p_n\geq 3$, we have $\frac{C_{I3DS}}{C_{OPT}}\leq 1.11$.  
\end{theorem}
\begin{proof}
We know that $Sum=27$ and initially $l_1=l_2=l_3=0$. We now have $C_{OPT}\geq 9$.
We can consider the following critical cases based on the size of job $J_1$\\
\textit{Case 1. If $p_1> \frac{10}{27}\cdot Sum$.}\\
We know that $p_{max}=p_1$. Clearly, $l_1+p_1> \frac{10}{27}\cdot Sum$ and $l_2+p_1> \frac{10}{27}\cdot Sum$. Let us consider $X=\sum_{i=2}^{n}{p_i}$. Clearly, $X< \frac{17}{27}\cdot Sum$. Algorithm \textit{I3DS} now assigns $J_1$ to machine $M_3$ and schedules the remaining jobs to machines $M_1$ and $M_2$ such that the updated loads $l_1\leq \frac{7}{27}\cdot Sum$ and $l_2\leq \frac{10}{27}\cdot Sum$, while $l_3=p_{max}> \frac{10}{27}\cdot Sum$. Hence, $C_{I3DS}=p_{max}$, while $C_{OPT}=p_{max}$. \\\\
\textit{Case 2. If $\frac{1}{3}\cdot Sum< p_1\leq \frac{10}{27}\cdot Sum $.}\\
Then $l_1+p_1> \frac{1}{3}\cdot Sum$ and $l_2+p_1\leq \frac{10}{27}\cdot Sum$. This implies, job $J_1$ is assigned to machine $M_2$ and the updated load $l_2\leq \frac{10}{27}\cdot Sum$. Now $\sum_{i=2}^{n}{p_i}\leq \frac{17}{27}\cdot Sum$. \\
\textit{Case 2(a). No other job than $J_1$ is assigned to machine $M_2$.}\\
We now can consider the following critical cases based on the size of the second job $J_2$\\
\textit{Sub-case 2.1. If $p_2=\frac{1}{3}\cdot Sum$.}\\
Then the updated load $l_1=\frac{1}{3}\cdot Sum$ and the remaining jobs are scheduled on machine $M_3$ such that the updated load $\frac{8}{27}\leq l_3< \frac{10}{27}\cdot Sum$. Hence $C_{I3DS}\leq \frac{10}{27}\cdot Sum$, while $C_{opt}\geq \frac{1}{3}\cdot Sum$. Therefore \\
\hspace*{2.2cm} $\frac{C_{I3DS}}{C_{OPT}}\leq \frac{10}{9}=1.11$ \hspace*{5.9cm}(14)\\
\textit{Sub-case 2.2. If $p_1<\frac{1}{3}\cdot Sum$.}\\
Then job $J_2$ is assigned to machine $M_1$. Now the updated load $l_1< \frac{1}{3}\cdot Sum$. We now can consider the following cases based on the assignment of jobs on machine $M_1$ as follows.\\
\textit{Sub-sub-case 2.2.1. No other job then $J_2$ is scheduled on machine $M_1$.}\\
Clearly, \hspace*{1.2cm} $l_1\leq \frac{1}{3}\cdot Sum-1$. \hspace*{5.9cm}(15)\\
Even the assignment of the last job $J_n$ to machine $M_1$ makes the updated load $l_1> \frac{1}{3}\cdot Sum$.\\
Before the assignment of job $J_n$ to machine $M_3$, we have $l_3\leq \frac{8}{27}\cdot Sum-1$. After the scheduling of job $J_n$ on $M_3$, we have the updated load $l_3>\frac{8}{27}\cdot Sum$ as $l_2\leq \frac{10}{27}\cdot Sum$ and $l_1< \frac{1}{3}\cdot Sum$. W.l.o.g, we can bound the processing time of the last job $J_n$ such that $2\leq p_n\leq p_{max}-6$. We now have $l_3\leq \frac{8}{27}\cdot Sum-1+p_{max}-6\leq \frac{8}{27}\cdot Sum+C_{OPT}-7\leq \frac{8}{27}\cdot Sum+\frac{1}{3}\cdot Sum-7\leq \frac{17}{27}\cdot Sum-7$, implies\\
\hspace*{2.7cm} $l_3\leq 10$ \hspace*{7.3cm} (16)\\
Hence $C_{I3DS}\leq 10$, while $C_{OPT}\geq 9$. Therefore\\
\hspace*{2.2cm} $\frac{C_{I3DS}}{C_{OPT}}\leq \frac{10}{9}=1.11$ \hspace*{5.9cm}(17)\\\\
\textit{Sub-sub-case 2.2.2. At least one job other than $J_2$ is assigned to machine $M_1$.}\\
Let $J_k$ be the latest job, which has been assigned to $M_1$. Let $l'_{1}$ be the load of $M_1$ just before the assignment of $J_k$. By Eq. (15), we have\\ 
\hspace*{2.2cm} $l'_{1}+p_n\leq \frac{1}{3}\cdot Sum> l_1$ \hspace*{5.1cm} (18)\\
If $J_n=J_k$, then by Eq. (18), we have $l'_{1}+p_n< \frac{1}{3}\cdot Sum$ or $l'_{1}+p_n= \frac{1}{3}\cdot Sum$. If $l'_{1}+p_n= \frac{1}{3}\cdot Sum$, then the updated load $l_3$ is such that $\frac{8}{27}\leq l_3< \frac{10}{27}\cdot Sum$, while $l_2\leq \frac{10}{27}\cdot Sum$. Hence $C_{I3DS}\leq \frac{10}{27}\cdot Sum$, while $C_{OPT}\geq \frac{1}{3}\cdot Sum$. Therefore $\frac{C_{I3DS}}{C_{OPT}}\leq \frac{10}{9}=1.11$.\\
If $l'_{1}+p_n< \frac{1}{3}\cdot Sum$. Let $l'_3$ be the load of machine $M_3$ after the scheduling of $J_n$ to machine $M_1$. By Eq. (16), we have $l'_3< l_3$ and $l'_3< \frac{10}{27}\cdot Sum$. Hence $C_{I3DS}=l_2\leq \frac{10}{27}\cdot Sum$, while $C_{OPT}\geq \frac{1}{3}\cdot Sum$. Therefore $\frac{C_{I3DS}}{C_{OPT}}\leq \frac{10}{9}=1.11$. \hfill\(\Box\) 
\end{proof}
\section{Conclusion and Future Work}
In this paper we investigated the non-preemptive semi-online scheduling problem in two and three identical parallel machines settings with makespan minimization objective. We studied the problem by considering a pair-wise combination of two types of \textit{EPI} on the future jobs, where we know \textit{jobs arrive in order of non-increasing sizes ($Decr$}) and the \textit{total sum of sizes of the jobs ($Sum$)} beforehand. In particular we studied the problem with respect to two practically significant input job sequence patterns, i.e., $I_1$ and $I_2$ based on $Decr$. The pattern $I_1$ represents a job sequence, where all jobs are of equal size and $I_2$ represents a job sequence, where jobs arrive in order of decreasing sizes. We showed that prior knowledge of multiple \textit{EPI} is extremely helpful in designing improved competitive semi-online scheduling algorithms. We proved the lower bounds $1.33$ and $1.11$ on the competitive ratio for the problem in two and three identical machines settings respectively by the adversary method. We proved a lower bound $1.04$ for the problem by considering only $I_2$ in two identical machines setting. Our method can further be extended to achieve a lower bound of $\frac{2m}{m+1}$ on the competitive ratio for the problem in $m$ identical machines settings, where $m\geq 2$.  We  proposed four new variants of the deterministic semi-online scheduling algorithms and achieved improved competitive ratios.  The summary of the upper bound results on the competitive ratio for our proposed semi-online algorithms is presented in Table \ref{tab:Summary of Our Results}.  
\begin{table}[!htbp]
\centering
\caption{Summary of Our Upper Bound Results for the problem $P|Decr, Sum|C_{max}$}
\begin{tabular} {cp{1.5cm}p{1.5cm}p{3cm}}
\hline
\textbf{Machines} & \textbf{$I_1$} &\textbf{$I_2$} & \textbf{$I_2$} \\
\hline
Two  & 1.33 &  1.33 & 1.16 \\
Three  &  1.5 &  1.2 & 1.11\\
\hline
\end{tabular}
\label{tab:Summary of Our Results}
\end{table} \\
\textbf{Future Work.} It is non-trivial and interesting to address the following research challenges as a part of future work.
\begin{itemize}
\item Improvement of the upper bound $1.16$ for the problem $P_2|Decr, Sum|C_{max}$ with $I_2$ for $n\geq 3$ and $p_i\geq 1$, $\forall i$.
\item Minimization of the upper bound $1.11$ for the problem $P_3|Decr, Sum|C_{max}$ with $I_2$ for $n\geq 4$ and $p_i\geq 1$, $\forall i$.
\item Study of the problem $P_m|Decr, Sum|C_{max}$ with respect to input patterns $I_1$ and $I_2$ for $m\geq 4$.
\item Exploration of practically significant new input patterns based on the  information on $Decr$ or job's release time.
\item As we know that \textit{EPI} is an application dependent parameter and it also affects the performance of a semi-online scheduling algorithm, it is interesting and challenging to explore practically and theoretically significant new \textit{EPI}.
\end{itemize} 
%

%
%
%
%

\begin{thebibliography}{00}
\bibitem{Lenstra:77}
Lenstra J K, Rinnooy Kan A H G, Brucker P. Complexity of machine scheduling problems. \textit{Annals of Discrete Mathematics}, 1977, 1:343-362.
\bibitem{Graham:69}
Graham R L. Bounds on multiprocessor timing anomalies. \textit{SIAM Journal on Applied Mathematics}, 1969,  17(2):416-429.
\bibitem{Beaumont:20}
Beaumont O, Canon L C, Eyraud-Dubois L, Lucarelli G, Marchal L, Mommessin C, Simon B, Trystram D. Scheduling on two types of resources: A Survey. \textit{ACM Computing Surveys (CSUR)}, 2020, 53(3):1-36.
\bibitem{Kellerer:97}
Kellerer H, Kotov V, Speranza M G, Tuza T. Semi-online algorithms for the partition problem. \textit{Operations Research Letters}, 1997, 21:235-242.
\bibitem{Seiden:00}
Seiden S, Sgall J, Woeginger G. Semi-online scheduling with decreasing job sizes. \textit{Operations Research Letters}, 2000, 27:215-221.
\bibitem{Tarjan:85}
Tarjan R E, Sleator D D. Amortized computational complexity. \textit{SIAM Journal on Algebraic and Discrete Methods}, 1985, 6(2):306-318.
\bibitem{Tan:02}
Tan Z, He Y. Semi-online problems on two identical machines with combined partial information. \textit{Operations Research Letters}, 2002, 30:408-414.
\bibitem{Angelelli:00}
Angelelli E. Semi-online scheduling on two parallel processors with known sum and lower bound on the size of the tasks. \textit{Central European Journal of Operations Research}, 2000, 8:285-295.
\bibitem{Angelelli:03}
Angelelli E, Speranza M G, Tuza Zs. Semi-online scheduling on two parallel processors with an upper bound on items. \textit{Algorithmica}, 2003, 37:243-262. 
\bibitem{Angelelli:06}
Angelelli E, Speranza M G, Tuza T. New bounds and algorithms for online scheduling: two identical processors, known sum and upper bound on the task. \textit{Discrete Mathematics and Theoretical Computer Science}, 2006, 8(1):1-16.
\bibitem{Cao:12}
Cao Q, Cheng T C E, Wan G, Li Y. Several semi-online scheduling problems on two identical machines with combined information. \textit{Theoretical Computer Science}, 2012, 457(26):35-44.
\bibitem{Cao:16}
Cao Q, Wan G. Semi-online scheduling with combined information on two identical machines in parallel. \textit{Journal of Combinatorial Optimization}, 2016, 31(2):686-695.
\bibitem{He:05}
He Y, Dosa G. Semi-online scheduling jobs with tightly grouped processing times on three identical machines. \textit{Discrete Applied Mathematics}, 2005, 150:140-159.
\bibitem{Angelelli:07}
Angelelli E, Speranza M G, Tuza Zs. Semi- online scheduling on three processors with known sum of the tasks. \textit{Journal of Scheduling}, 2007, 10:263-269.
\bibitem{Hua:06}
Hua R, Hu J, Lu L. A semi-online algorithm for parallel machine scheduling on three machines. \textit{Journal of Industrial Engineering and, Engineering Management}, 2006, 20.
\bibitem{Wu:07}
Wu Y, Tan Z, Yang Q. Optimal semi-online scheduling algorithms on a small number of machines. \textit{Lecture Notes in Computer Science (LNCS)}, Springer, 2007, 4614:504-515.
\bibitem{Cheng:12}
Cheng T C E, Kellerer H, Kotov V. Algorithms better than LPT for semi-online scheduling with decreasing processing times. \textit{Operations Research Letters}, 2012, 40:349-352.
\bibitem{Epstein:18}
Epstein L. A Survey on makespan minimization in semi-online environments. \textit{Journal of Scheduling}, 2018, 21(3):269-284.
\bibitem{Graham:66}
Graham R L. Bounds for certain multiprocessor anomalies. \textit{Bell System Technical Journal}, 1966, 45(1):1563--1581.
\end{thebibliography}
\end{document}